\documentclass[11pt]{article} 

\usepackage[utf8]{inputenc} 

\usepackage{geometry} 
\usepackage{graphicx} 

\usepackage{booktabs} 
\usepackage{array} 
\usepackage{paralist} 
\usepackage{verbatim} 
\usepackage{subfig} 
\usepackage{natbib}
\usepackage{hyperref}
\usepackage{tikz}

\usepackage{fancyhdr} 
\usepackage{sectsty}
\allsectionsfont{\sffamily\mdseries\upshape} 

\usepackage[nottoc,notlof,notlot]{tocbibind} 
\usepackage[titles,subfigure]{tocloft} 

\usepackage{amsmath}
\usepackage{amssymb}
\usepackage{amsthm}
\usepackage{mathrsfs}
\usepackage{enumerate}
\usepackage[all]{xy}
\usepackage{todonotes}
\usetikzlibrary{calc}

\newtheorem{thm}{Theorem}[subsection]
\newtheorem{prop}[thm]{Proposition}
\newtheorem{cor}{Corollary}[thm]

\newtheorem{lem}[thm]{Lemma}
\theoremstyle{definition}

\newtheorem{defn}[thm]{Definition}

\newcommand{\ol}[1]{\overline{#1}}
\newcommand{\wt}[1]{\widetilde{#1}}
\newcommand{\wh}[1]{\widehat{#1}}
\newcommand{\vocab}[1]{\textbf{#1}}

\newcommand{\bbr}{\mathbb R}
\newcommand{\bbz}{\mathbb Z}

\newcommand{\bbn}{\mathbb N}

\newcommand{\bbc}{\mathbb C}

\newcommand{\diff}{\backslash}

\renewcommand{\Im}{\text{Im}}

\newcommand{\adots}{\reflectbox{$\ddots$}}

\newcommand{\sheaf}[1]{\mathcal{#1}}

\newcommand{\supp}{\textnormal{supp}}

\newcommand{\grad}{\nabla}
\newcommand{\Ro}{\text{Ro}}
\newcommand{\Fr}{\text{Fr}}
\newcommand{\Bu}{\text{Bu}}

\title{A Connection Between Orthogonal Polynomials and Shear Instabilities in the Quasi-geostrophic Shallow Water Equations}
\author{W.R. Casper}

\begin{document}

\maketitle
\begin{abstract}
In this paper we demonstrate a connection between the roots of a certain sequence of orthogonal polynomials on the real line and the linear instability of a $x$-directionally homogeneous background velocity profile $u^b(x,y) = \cos(y)$ in the quasi-geostrophic shallow water (QG) equation in a domain with periodic boundaries in the $y$-direction.  Using the relationship we establish, we then prove that there exists a unique unstable mode for each horizontal wave number $0<k<1$ and provide mathematically rigorous estimates of the associated growth rate.
\end{abstract}

\section{Introduction}
In this paper, we obtain for each wave number $k>0$ rigorous bounds for the eigenvalues of the (modified) Rayleigh stability equation
\begin{equation}\label{Rayleigh equation}
f''(y)-\left(k^2 + \frac{u^b(y)''- c/\Bu}{u^b(y)-c}\right)f(y) = 0.
\end{equation}
with periodic boundary conditions $f(y + 2\pi) = f(y)$ in the case that the background velocity field $u^b(y) = \cos(y)$.  Here by eigenvalues, we refer to the values of $c$ (for fixed $k$) for which the Rayleigh equation has a square-integrable, periodic solution.

The modified version of Rayleigh's equation above arises in the study of the linear stability of shear flows in the inviscid, incompressible shallow water equations in the limit of small Rossby number $\Ro$.  The quasi-geostrophic shallow water equation (QG) is given by
\begin{equation}\label{QG equation}
q_t + \psi_xq_y - \psi_yq_x = 0
\end{equation}
where $\psi = \psi(x,y,t)$ is the stream function and $q = q(x,y,t)$ is the potential vorticity, related by $q = \Delta\psi - \psi/\Bu$.  Linearizing the QG equation around a shear background stream function $\psi^b = \psi^b(y)$, we obtain the linear partial differential equation for the perturbed stream function $\psi^p$
$$\left(\Delta-\frac{1}{\Bu}\right)\psi^p_t + (\psi^b)'''\psi_x^p - (\psi^b)'\Delta\psi_x^p = 0.$$
Note that the background stream function determines a background velocity $u^b = - (\psi^b)'$.  If $f(y)$ is a solution to Equation \ref{Rayleigh equation} with this value of $u^b$ for some pair $(k,c)$, then $\psi^p(x,y,t) = e^{ik(x-ct)}f(y)$ is a solution of the linearized QG equation.  We derive the QG equation from the shallow water equation in the limit of small $\Ro$ in Appendix \ref{QG derivation} below.  We derive the linearized QG equations and the Rayleigh equation in Appendix \ref{Rayleigh derivation}.

For a given value of $k$, there will in general be countably (and often finitely many) values of $c$ for which Equation \ref{Rayleigh equation} will have a square-integrable, periodic solution (ie. for which $c$ is an eigenvalue).  In this way the choice of background profile $u^b(y)$ determines a dispersion relation, ie. a relationship between (complex) frequencies $c$ and wave numbers $k$, which we can represent as a multi-valued function $c(k)$.  Note that if $c$ is a value of $c(k)$, then so too is $\ol c$.  A wave number $k$ is called unstable if one of the values of $c(k)$ is nonreal, and in this case an associated perturbed solution $\psi^p$ of the linear QG equation grows exponentially with growth rate $k\cdot\Im(c(k))$.

For some very special background velocity profiles $u^b(y)$, the solutions of Equation \ref{Rayleigh equation} may be determined explicitly analytically and the dispersion relation thereby determined also.  However, for the vast majority of profiles this is not the case.  Instead, numerical methods of determining the dispersion relation $c(k)$ are required.  One popular method is to replace the differential operators in Rayleigh's equation with approximations in the form of finite-dimensional linear operators acting on a finite-dimensional vector space.  To do so, we can replace the interval $[0,2\pi]$ with a finite grid, and the differential operators with difference operators on this grid \cite{wang2012ageostrophic}\cite{menesguen2012ageostrophic}\cite{gula2010instabilities}.  Alternatively, we can expand $f$ in terms of a orthonormal basis for $[0,2\pi]$ and take a finite truncation \cite{dolph1958application}\cite{gallagher1962behaviour}\cite{orszag1971accurate}.  Either way, this replaces Equation \ref{Rayleigh equation} with a simple eigenvalue problem on a finite-dimensional vector space, and we can imagine that as the accuracy of our approximation is increased that the scattering relations obtained by the various approximations will converge to the true scattering relation $c(k)$.

This presents us with a problem.  We have to try to tell which of the eigenvalues of the various approximations are also approximations of the eigenvalues of Equation \ref{Rayleigh equation}.  This problem becomes even more apparent for the wide class of background velocity profiles $u^b(y)$ for which Equation \ref{Rayleigh equation} has finitely many eigenvalues for each fixed value of $k$.  As the precision of our approximations increases, so too does the dimension of the linear system approximating Rayleigh's equation, resulting in an ever increasing amount of eigenvalues.  Even worse, we have no explicit estimates of the rate of convergence.

It is useful to rephrase this problem in the language of differential operators.  Consider the Schr\"{o}dinger operator $L_c$ which acts as an unbounded operator on (a dense subset of) the Hilbert space $\sheaf H = L^2([0,2\pi])$ of square-integrable functions on the interval $[0,2\pi]$ by
\begin{equation}\label{Schroodinger operator}
L_c[f] := f''(y) - \left(\frac{u^b(y)'' - c/\Bu}{u^b(y)-c}\right)f(y).
\end{equation}
for all $f$ in the domain $\sheaf D(L_c)$ of $L_c$
$$\sheaf D(L_c) := \left\lbrace f\in\sheaf H: \text{$f''(y)$ exists and }f''(y) - \left(\frac{u^b(y)'' - c/\Bu}{u^b(y)-c}\right)f(y)\in\sheaf H\right\rbrace.$$

The spectrum of any unbounded operator $L$ on a Hilbert space $\sheaf H$ with domain $\sheaf D(L)$ is composed of three parts: a discrete component $\sigma_d(L)$, an continuous component $\sigma_c(L)$, and a singular component $\sigma_s(L)$.  Points in each component of the spectrum are characterized as follows:
\begin{align*}\label{spectral decomposition}
\sigma_d(L) &= \{\lambda\in \bbc: \text{$L-\lambda I$ is not injective}\}\\
\sigma_c(L) &= \{\lambda\in \bbc: \text{$L-\lambda I$ is injective, $\sheaf R(L-\lambda I)$ is a dense, proper subset of $\sheaf H$}\}\\
\sigma_s(L) &= \{\lambda\in \bbc: \text{$L-\lambda I$ is injective,  $\sheaf R(L-\lambda I)$ is not dense in $\sheaf H$}\}
\end{align*}
where in the above $\sheaf R(L-\lambda I)$ denotes the range of $L-\lambda I$.  With this in mind, the question of determining the eigenvalues of \ref{Rayleigh equation} for all values of $k$ is equivalent to determining for which values of $c$ the operator $L_c$ has a positive eigenvalue $k^2$ in its discrete spectrum.

In this paper, we consider the shear background profile $u^b(y)=\cos(y)$.  This profile is complicated enough that the solution to Equation \ref{Rayleigh equation} cannot be obtained analytically.  However, we will show that explicit and rigorous estimates for the dispersion relation above can be made.  Our method for estimating the eigenvalues of Rayleigh's equation \ref{Rayleigh equation} for a background cosine profile is based on relating the eigenvalues to the roots of a sequence of orthogonal polynomials for a certain measure defined on the real line $\bbr$.  Roots of orthogonal polynomials satisfy an interlacing property which we apply to determine the number of eigenvalues and obtain monotonic sequences whose limits are the desired eigenvalues.  Specifically we prove the following
\begin{thm}\label{main theorem}
Let $b_0 = 2z_0^2 + z_1^2$, $b_j = z_{2j}^2 + z_{2j+1}^2$ for $j\geq 1$, and $a_j = z_{2j+1}z_{2j+2}$ for all $j\geq 0$, where
$$z_j = \frac{1}{2}\left(\frac{(j^2+k^2-1)((j+1)^2+k^2-1)}{(j^2+k^2 + 1/\Bu)((j+1)^2+k^2 + 1/\Bu)}\right)^{1/2},$$
and let $p_{-1}(x)=0,p_0(x) = 1$ and define $p_{n+1}(x)$ recursively for $n\geq 0$ by
$$xp_n(x) = a_np_{n+1}(x) + b_np_n(x) + a_{n-1}p_{n-1}(x).$$
Take $u^b(y)=\cos(y)$ in Rayleigh's equation \ref{Rayleigh equation}.  Then the following holds
\begin{enumerate}[(a)]
\item  Rayleigh's equation \ref{Rayleigh equation} has complex eigenvalues if and only if $|k|<1$
\end{enumerate}
Furthermore, assuming $|k|<1$, the following is true
\begin{enumerate}[(a)]
\setcounter{enumi}{1}
\item  For $n$ large enough, the polynomial $p_n(x)$ has a unique negative root $-r_n$.
\item  The sequence of positive real number $r_1,r_2,\dots$ from (b) is monotone increasing and converges to a real number $r$.
\item  The complex eigenvalues of Rayleigh's equation \ref{Rayleigh equation} are given by $\pm i\sqrt{r}$.
\end{enumerate}
\end{thm}
The most significant part of Theorem \ref{main theorem} is that the roots of the polynomials $p_n(x)$ converge \emph{monotonically}.  Therefore for each $n$, we get a new, sharper lower bound for the growth rate of instabilities.

\section{General Results on Rayleigh's Equation}
Before diving into some of the mathematical background on orthogonal polynomials used in this paper, it makes sense to recount some of the more basic results known for the Rayleigh equation.  However, since Equation \ref{Rayleigh equation} is not exactly the Rayleigh equation, but a modified version, these results will change in some important ways.

Many general results on Rayleigh's equation involve finding criteria for the existence of non-real eigenvlaues, and bounds for the growth rates of the associated unstable linear modes.  We are dealing with a modified version of Rayleigh's equation however, because of the presence of the Burger number factor, and this modifies some of the usual instability criteria in interesting ways.  For example, a well-known criterion for instability is Rayleigh's inflection point criterion, which says that for a smooth, shear profile to be linearly unstable, it must have an inflection point.  We see in the next theorem, that we no longer need an inflection point if the Burger number is small enough.
\begin{thm}[Rayleigh's Inflection Point Criterion]
Let $u^b(y)$ be twice differentiable with continuous second derivative.  Suppose that for fixed $k$, Equation \ref{Rayleigh equation} has a non-real eigenvalue $c$.  Then there exists a point $y_0$ satisfying
$$u^b(y_0)''-u^b(y_0)/\Bu = 0.$$
\end{thm}
\begin{proof}
Let $f(y)$ be an eigenvector for the eigenvalue $c$.  Then multiplying Equation \ref{Rayleigh equation} by $f^*(y)$, integrating by parts, and taking the imaginary part of the resultant identity, we find
$$c_i\int_0^{2\pi} \frac{u^b(y)''-u^b(y)/\Bu}{(u^b(y)-c_r)^2 + c_i^2}|f(y)|^2 dy = 0.$$
Since $c_i\neq 0$, the statement of our theorem is follows from the intermediate value theorem.
\end{proof}

We show in this paper that $u^b(y) = \cos(y)$ is an unstable profile.  Yet, by the above theorem $u^b(y) = \cos(y) + 2$ is a stable profile for $\Bu < 1$.  This is indicative of an important difference between Rayleigh's original equation and Equation \ref{Rayleigh equation}, namely that the choice of inertial reference frame matters.  This is a consequence of the fact that the modified equation was derived in a rotating reference frame.  We note that as the rate of rotation is decreased to $0$, the Burger number increases to $\infty$ and Rayleigh's equation takes it's traditional form; in particular the choice of reference frame no longer matters.

Another well-known result is Howard's semicircle theorem, which states that for a background velocity profile taking values in the finite interval $[u_{min},u_{max}]$, the eigenvalues of the equation corresponding to unstable modes occur in a circle of radius $(u_{max}-u_{min})/2$ centered at $(u_{max}+u_{min})/2$.  However, it seems to be the case that Howard's semicircle theorem as stated does not hold for the modified Rayleigh equation.  This is again due to the fact that the choice of reference frame matters.  Instead, we have a modified form, which turns out to be equivalent to Howard's semicircle theorem in the limit $\Bu\rightarrow\infty$.
\begin{thm}[Centered Howard's Semicircle Theorem]
Let $u^b(y)$ be twice differentiable with continuous second derivative, with $u^b(y)\in [-r,r]$.  Suppose furthermore that for fixed $k$, Equation \ref{Rayleigh equation} has an eigenvalue $c$.  Then $c$ lies in the complex plane within a circle of radius $r$ of the origin.
\end{thm}
\begin{proof}
Make the substitution $g(y) = f/(u^b-c)$.  With this, Rayleigh's equation becomes
$$((u^b-c)^2g'(y))' - \left(k^2(u^b-c)^2-\frac{c(u^b-c)}{\Bu}\right)g(y) = 0.$$
Multiplying both sides of this equation by $g^*(y)$ and integrating by parts, we obtain
$$-\int_0^{2\pi} (u^b-c)^2 (|g'(y)|^2+k^2|g(y)|^2)dy + \frac{c}{\Bu}\int_0^{2\pi}(u^b-c)|g(y)|^2dy = 0.$$
Taking real and imaginary parts and simplifying with $Q(y) = |g'(y)|^2 + (k^2+1/\Bu)|g(y)|^2$, we obtain
$$\int_0^{2\pi} u^b(y)Q(y)dy = c_r\int_{0}^{2\pi}Q(y)dy +  \frac{1}{2\Bu}\int_0^{2\pi}u^b|g(y)|^2dy$$
$$\int_0^{2\pi} u^b(y)^2Q(y)dy = (c_r^2+c_i^2)\int_0^{2\pi}Q(y)dy + \frac{1}{\Bu}\int_0^{2\pi}u^b(y)^2|g(y)|^2dy$$
where we have written $c$ in terms of its real and imaginary components $c = c_r + ic_i$.
This latter equation in particular says that 
$$(r^2-|c|^2)\int_0^{2\pi} Q(y)dy \geq \frac{1}{\Bu}\int_0^{2\pi}u^b(y)^2|g(y)|^2dy$$
In particular, $|c|^2\leq r^2$.
\end{proof}

\section{Jacobi Matrices, Polynomials, and Measures}
Next, we will provide a brief summary of the relation between orthogonal matrix polynomials and Jacobi matrices.
\begin{defn}
A \vocab{Jacobi matrix} is an infinite or finite square, tri-diagonal matrix of the form
\begin{equation}\label{Jacobi matrix}
\left[\begin{array}{ccccc}
b_0 & a_0 &  0  &  0  & \dots\\
a_0 & b_1 & a_1 &  0  & \dots\\
 0  & a_1 & b_2 & a_2 & \dots\\
 0  &  0  & a_2 & b_2 & \dots\\
 \vdots &  \vdots  & \vdots & \vdots & \ddots\\
\end{array}\right]
\end{equation}
for some sequences of complex numbers $a_i,b_i\in\bbc$ with $a_i\neq 0$ for all $i$.
\end{defn}
\begin{defn}
Let $A$ be an $n\times n$ Jacobi matrix of the form \ref{Jacobi matrix}.  The \vocab{associated sequence of polynomials} is the sequence $p_0(x), p_1(x),\dots$ defined recursively by $p_0(x) = 1$, $p_1(x) = a_0^{-1}(x-b_0)$, and
$$xp_n(x) = a_np_{n+1}(x) + b_np_n(x) + a_{n-1}p_n(x),\ n\geq 1.$$
\end{defn}
The roots of the associated sequence of polynomials describe the eigenvalues of a finite Jacobi matrix, as stated in the next proposition.  This is proved in several places, including \cite{totik2005orthogonal}.
\begin{prop}
Let $A$ be an $n\times n$ Jacobi matrix of the form \ref{Jacobi matrix}, and let $p_0(x),\dots, p_n(x)$ be the associated sequence of polynomials.  Then the eigenvalues of $A$ are the roots $\lambda_1,\dots,\lambda_n$ of the polynomial $p_n(x)$.  The corresponding eigenspaces are given by
$$\sheaf E_{A}(\lambda_i) := \text{span}\{[p_0(\lambda_j),p_1(\lambda_j),\dots, p_{n-1}(\lambda_j)]^T\}.$$
\end{prop}
The spectrum $\sigma(A)$ of $A$ for infinite Jacobi matrices is more complicated, and consists of both a discrete part $\sigma_d(A)$, a continuous part $\sigma_c(A)$ and a singular part $\sigma_s(A)$.  Often the singular part of the spectrum of $A$ is empty, as is the case when $A$ is essentially normal.
\begin{prop}[\cite{arlinskiui2006non}]
Let $A$ be an infinite Jacobi matrix with bounded coefficients.  Then $A$ defines a bounded linear operator on the Hilbert space $\ell^2(\bbn)$ whose spectrum $\sigma(A)$ consists of limit points of $\bigcup_n \{\lambda: p_n(\lambda) = 0\}$.  The discrete component of the spectrum is
$$\sigma_d(A) = \left\lbrace\lambda: \sum_{n=1}^\infty |p_n(\lambda)|^2 < \infty\right\rbrace.$$
\end{prop}
There is a correspondence between Hermitian Jacobi matrices and probability measures on the real line.  Under this correspondence, the support of the measure agrees with the spectrum of the Jacobi matrix.  This result is often called Favard's Theorem and is proved in \cite{favard1935polynomes} but was also proved by others, including Stieltjes in \cite{stieltjes1894recherches}.
\begin{thm}[Favard's Theorem\cite{favard1935polynomes}\cite{stieltjes1894recherches}]
Let $A$ be a Jacobi operator with real, bounded coefficients, and let $p_0(x),p_1(x),\dots$ be the associated sequence of polynomials.  Then $\sigma(A)$ is a bounded subset of $\bbr$ and there exists a positive measure $\mu$ supported on $\sigma(A)$ satisfying
$$\int |x|^nd\mu(x) < \infty$$
and also
$$\int p_j(x)p_k(x)d\mu(x)=0\ \text{for all $j,k$ with $j\neq k$.}$$
\end{thm}
A sequence of polynomials $p_0(x),p_1(x),\dots$ with $\deg(p_i) = i$ for all $i$, satisfying the identity $\int p_j(x)p_k(x)d\mu(x)$ is called a sequence of orthogonal polynomials for the measure $\mu$.

The correspondence between Hermitian Jacobi matrices and the associated probability measures is clarified even further by considering the Stieltjes transform of the measure.
\begin{prop}[\cite{gesztesy1997m}]
Let $A$ be a Jacobi operator with real, bounded coefficients.  Consider the Stieltjes transform of $\mu(x)$
$$m(z) := \int_{\sigma(A)} \frac{1}{x-z} d\mu(x),\ z\notin \sigma(A).$$
Then $m(z)$ is a meromorphic function on $\bbc\diff\sigma_c(A)$ for which the following is true
\begin{enumerate}[(a)]
\item  the set of poles of $m(z)$ is $\sigma_d(A)$
\item  $m(z)$ has the Laurent series expansion
$$m(z) = -\sum_j \frac{m_j}{z^{j+1}},\ \ \text{where}\ \ m_j = \int x^j d\mu(x).$$
\end{enumerate}
If additionally $a_j > 0$ for all $j$, then
\begin{enumerate}[(a)]
\setcounter{enumi}{2}
\item  if $a_j>0$ for all $j$, then $m(z)$ has the continued fraction expansion
$$m(z) = \frac{-1}{z-b_1 + a_1^2\left[\frac{-1}{z-b_2 + a_2^2[\dots]}\right]}.$$
\end{enumerate}
\end{prop}
\section{Orthogonal Polynomials, Root Interlacing, and Growth Rates}
\subsection{Orthogonal Polynomials}
\begin{defn}
Let $\mu$ be a real, positive measure on the real line.  We say that \vocab{$\mu$ has finite} moments if $\int_\bbr |x|^n d\mu(x) < \infty$ for all $n\geq 0$.
\end{defn}
A positive measure with finite moments $\mu$ defines a real inner product on the vector space of polynomials $\bbr[x]$ via the formula
$$\langle p(x),q(x)\rangle_\mu := \int_\bbr p(x)q(x)d\mu(x).$$
By Gram-Schmidt orthogonalization, we may construct an orthogonal basis $p_0(x), p_1(x),\dots$ for $\bbr[x]$ such that $\deg(p_n) = n$ for all $n\geq 0$.
\begin{defn}
Let $\mu$ be a real, positive measure on the real line with finite moments.  A \vocab{sequence of orthogonal polynomials} for $\mu$ is a sequence of polynomials $p_0(x),p_1(x),\dots$ satisfying $\deg(p_n) = n$ for all $n\geq0$ and $\int p_m(x)p_n(x)d\mu(x) = 0$ for all pairs $m,n$ with $m\neq n$.
\end{defn}
If $p_0(x),p_1(x),\dots$ and $q_0(x),q_1(x),\dots$ are two seqeuences of orthogonal polynomials for $\mu$, then there exists constants $c_0,c_1,\dots$ such that $p_i(x) = c_iq_i(x)$ for all $i$.  Thus orthogonal polynomials are essentially unique for a given measure.

The converse of Favard's theorem also holds.  Given a probability measure and a sequence of (normalized) orthogonal polynomials, one may prove that the polynomials satisfy a $3$-term recursion relation, ie. they are the same as the associated polynomials of some Jacobi matrix.  This is proven in many places, including \cite{totik2005orthogonal}.
\begin{prop}[\cite{totik2005orthogonal}]
Let $\mu$ be a real, positive measure on the real line with finite moments.  Then there exists a Jordan matrix $J$ whose associated sequence of orthogonal polynomials is a sequence of orthogonal polynomials for $\mu$.
\end{prop}

\subsection{Root Interlacing}
Suppose that $p_0(x),p_1(x),\dots$ are the orthogonal polynomials for some positive measure $\mu$ on the real line $\bbr$ with finite moments, and whose support has infinite cardinality.  Then $p_n(x)$ has $n$ distinct roots for each integer $n$ and between any two roots of $p_{n+1}(x)$ there must lie a root of $p_n(x)$.
\begin{prop}[\cite{nevai1989orthogonal}]\label{root interlacing}
Let $\mu$ be a positive measure on the real line with finite moments, and suppose the cardinality of $\supp(\mu)$ is not finite.  Let $p_0(x),p_1(x),\dots$ be a sequence of orthogonal polynomials for $\mu$.  Then the following is true
\begin{enumerate}[(a)]
\item  $p_n(x)$ has $n$ distinct, real roots $r_{n,1} < r_{n,2} < \dots < r_{n,n}$ for all positive integers $n$
\item  the roots of the polynomials satisfy the interlacing property for all $n\geq 1$:
$$r_{(n+1),1} < r_{n,1} < r_{(n+1),2} < r_{n,2} < \dots < r_{n,n} < r_{(n+1),n}.$$
\end{enumerate}
\end{prop}

As a corollary of this, we note that between any two roots of $p_n(x)$ there must exist an accumulation point of the spectrum of $J$.
\begin{cor}\label{interlacing corollary}
Let $p_0(x),p_1(x),\dots$ be a sequence of orthogonal polynomials for a positive measure $\mu$ on the real line whose support has infinte cardinality.  Let $r_{n,j}$ be defined as in the statement of the proposition for all $n>0$ and all $1\leq j\leq n$.  Then for all $1\leq j < n$ and $n\geq 2$ we have that
$$\sigma(A)_\ell\cap [r_{n,j},r_{n,{j+1}}]\neq 0,$$
where here $A$ is the associated Hermitian Jacobi matrix and $\sigma(A)_\ell$ denotes the set of limit points of the spectrum $\sigma(A)$ of $A$.
\end{cor}
\begin{proof}
As a consequence of the interlacing property, the set
$$[r_{n,j},r_{n,j+1}]\cap \bigcup_{n\geq 1}\{\lambda: p_n(\lambda)=0\}.$$
has infinite cardinality.  Since $[r_{n,j},r_{n,j+1}]$ is compact, it follows that $\bigcup_{n\geq 1}\{\lambda: p_n(\lambda)=0\}$ has a limit point in $[r_{n,j},r_{n,j+1}]$.  Hence $[r_{n,j},r_{n,j+1}]\cap \sigma(A)$ is nonempty.  This argument applied again to the successive root pairs in $[r_{n,j}, r_{n,j+1}]$ actually shows that $[r_{n,j},r_{n,j+1}]\cap \sigma(A)$ has infinitely many points.  Hence it has a limit point.
\end{proof}

\subsection{Growth Rates}
We will also require estimates for the growth rates of $p_n(x)$ for $x\notin\supp(\mu)$.  The main idea is that if $x\notin\supp(\mu)$, then the magnitude of $p_n(x)$ grows exponentially in $n$.  The following result was obtained by Brian Simanek based on results of Barry Simon \cite{simanek2011new}\cite{simon2004orthogonal}
\begin{thm}[\cite{simanek2011new}]\label{growth rate}
Let $A$ be a Jacobi matrix of the form \ref{Jacobi matrix}, and let $\mu(x)$ and $p_0(x),p_1(x),\dots$ be the associated measure and sequence of orthogonal polynomials.  Suppose that $\mu$ has compact support on the real line, and moreover that
$$\lim_{n\rightarrow\infty} a_n = r,\ \ \lim_{n\rightarrow\infty} b_n = x_0.$$
Then for all $z\notin \supp(\mu)$
$$\lim_{n\rightarrow\infty} \frac{p_n(z)}{p_{n-1}(z)} = \frac{(z-x_0)/r + \sqrt{(z-x_0)^2/r^2-4}}{2}.$$
\end{thm}
If $A$, $\mu$, and $p_0(x),p_1(x),\dots$ satisfy the assumptions of Theorem (\ref{growth rate}), then the support of the absolutely continuous component of $\mu$ is contained in $[x_0-2r, x_0+2r]$.  For real points $z$ outside this interval, and outside the support of $\mu$, the magnitude of $\frac{(z-x_0)/r + \sqrt{(z-x_0)^2/r^2-4}}{2}$ is greater than $1$, and therefore the magnitude of $p_n(z)$ grows exponentially fast in $n$ for large $n$.

\section{Rayleigh's Equation and Jacobi Matrices}
Suppose that we wish to find square-integrable solutions of Rayleigh's equation \ref{Rayleigh equation} for given $k$.  Any such solution $f(y)$ has a Fourier expansion:
$$f(y) = \sum_{\ell=-\infty}^\infty \hat f(\ell)e^{i\ell y},$$
and inserting this into Rayleigh's equation along with the Fourier expansion of $u^b(y)$ and simplifying yeilds the following eigenvalue problem for $u^b(y)$:
\begin{equation}\label{eigenvalue problem}
\left(\ell^2+k^2+\frac{1}{\Bu}\right)^{-1}\left(\hat{u^b}(\ell)*((\ell^2+k^2)\hat f(\ell)) - (\ell^2\hat{u^b}(\ell))*\hat f(\ell)\right) = c\hat f(\ell),
\end{equation}
where here $*$ denotes the (discrete) convolution operator.
\subsection{Instability of the Cosine Profile}
We next consider specifically the case that $u^b(y) = \cos(y)$.  In this case, the eigenvalue problem \ref{eigenvalue problem} becomes
\begin{equation}\label{cosine problem}
\frac{1}{2}\frac{(\ell+1)^2+k^2 - 1}{\ell^2 + k^2 + 1/\Bu}\hat f(\ell+1) + \frac{1}{2}\frac{(\ell-1)^2+k^2 - 1}{\ell^2 + k^2 + 1/\Bu}\hat f(\ell-1)= c\hat f(\ell),
\end{equation}
Setting
$$\hat q(\ell) = \frac{1}{\sqrt{2}}\left(\frac{\ell^2 + k^2 - 1}{\ell^2 + k^2 + 1/\Bu}\right)^{1/2}\ \ \text{and}\ \ \hat g(\ell) = (\ell^2 + k^2 - 1)^{1/2}(\ell^2 + k^2 + 1/\Bu)^{1/2}\hat f(\ell)$$
we find
\begin{equation}\label{modified cosine problem}
c\hat g(\ell) = \hat q(\ell)\hat q(\ell+1)\hat g(\ell+1) + \hat q(\ell)\hat q(\ell-1) \hat g(\ell-1).
\end{equation}
Hence $[\dots,\hat g(-2),\hat g(-1),\hat g(0),\hat g(1),\dots]^T$ is an eigenvector with eigenvalue $c$ of the bi-infinite tri-diagonal matrix
$$B = \left[\begin{array}{ccccccc}
\ddots  & \vdots & \vdots & \vdots & \vdots & \vdots & \adots\\
\dots   & 0 & z_{1} & 0 & 0 & 0 & \dots\\
\dots   & z_{1}  &   0    & z_0&  0  &  0  & \dots\\
\dots   & 0 &  z_0   &  0  & z_0 &  0  & \dots\\
\dots   & 0 &   0    & z_0 &  0  & z_1 & \dots\\
\dots   & 0 &  0  &  0  & z_1 &  0  & \dots\\
 \adots & \vdots &  \vdots  & \vdots & \vdots &  \vdots & \ddots\\
\end{array}\right],\ \ \text{for}\ z_j = \wh q(j)\wh q(j+1).$$
Here we have used the fact that $\hat q(j) = \hat q(-j)$, and that therefore $z_{-j} = z_{j-1}$ for $j<0$.  Note that $B$ is \emph{not} a Jacobi matrix, because it is bi-infinite.

We next show how to relate the eigenvalues of $B$ to the eigenvalues of an infinite Hermitian Jacobi matrix.  The symmetry of $B$ implies that each eigenspace is invariant under the involution
$$\sigma([\dots,v(-1),v(0),v(1),\dots]) = [\dots,v(1),v(0),v(-1),\dots].$$
Therefore $B$ must have an eigenvector $\vec v$ with eigenvalue $c$ satisfying $\sigma(\vec v) = \pm\vec v$, ie. each eigenvalue of $B$ must have either a $\sigma$-symmetric or $\sigma$-skew symmetric eigenvector.  If $\sigma(\vec v) = \vec v$, then $cv(0) = 2z_0v(1)$ and therefore $[v(0),v(1),v(2),\dots]^T$ is an eigenvector of 
$$\wt B = \left[\begin{array}{ccccc}
 0  & 2z_0 &  0  &  0  & \dots\\
z_0 &  0  & z_1 &  0  & \dots\\
 0  & z_1 &  0  & z_2 & \dots\\
 0  &  0  & z_2 &  0  & \dots\\
 \vdots &  \vdots  & \vdots & \vdots &  \ddots\\
\end{array}\right],\ \ \text{for}\ z_j = \wh q(j)\wh q(j+1)$$
with eigenvalue $c$.  If we square $\wt B$ and take the imaginary part, then we see that $[v(0),v(1),v(2),\dots]^T$ is an eigenvector of 
$$\wt B^2 = \left[\begin{array}{ccccc}
 2z_0^2  & 0  &  2z_0z_1  &  0  & \dots\\
 0 &  2z_0^2+z_1^2  &  0 &  z_1z_2  & \dots\\
 z_0z_1  & 0  & z_2^2+z_2^2 & 0 & \dots\\
 0  &  z_1z_2  & 0 &  z_2^2+z_3^2  & \dots\\
 \vdots &  \vdots  & \vdots & \vdots &  \ddots\\
\end{array}\right].$$
It follows from the above checkerboard pattern that $[v(1),v(3),v(5),\dots]^T$ is an eigenvector with eigenvalue $c^2$ of the infinite Hermitian Jacobi matrix
$$A=
\left[\begin{array}{ccccc}
2z_0^2+z_1^2 & z_1z_2       &  0          &          0  & \dots\\
z_1z_2       & z_2^2+z_3^2  & z_3z_4      &  0          & \dots\\
 0           & z_3z_4       & z_4^2+z_5^2 & z_5z_6      & \dots\\
 0           &  0           & z_5z_6      & z_6^2+z_7^2 & \dots\\
 \vdots &  \vdots  & \vdots & \vdots & \ddots\\
\end{array}\right]
$$

The $\sigma$-symmetric eigenvectors are significant, because the unstable modes exhibit this symmetry.
\begin{lem}
Suppose that $c$ is a non-real eigenvalue of Equation \ref{cosine problem}.  Then the associated eigenspace of $B$ is $1$-dimensional and consists of a single $\sigma$-symmetric vector.
\end{lem}
\begin{proof}
Suppose that $c$ is a non-real eigenvalue of Equation \ref{cosine problem}, and let $\vec v = [\dots, v(-1),v(0),v(1),\dots]$ be an eigenvector of $B$ with this eigenvalue.  If $v(0) = 0$, then $[v(1),v(2),v(3),\dots]$ is an eigenvector with eigenvalue $c$ of the infinite Hermitian Jacobi matrix
$$J = \left[\begin{array}{ccccccc}
0     &  z_{1} & 0   &  0  & \dots\\
z_{1} &   0    & z_2 &  0  & \dots\\
0     &  z_2   &  0  & z_3 & \dots\\
0     &   0    & z_3 &  0  & \dots\\
\vdots &  \vdots  & \vdots & \vdots &  \ddots\\
\end{array}\right].$$
However, this implies that $c$ is real, which is a contradiction.  Therefore $v(0)\neq 0$.  If the eigenspace with eigenvalue $c$ contains more than one linearly independent vector, then by taking an appropriate linear combination we can obtain a nonzero eigenvector with $v(0)=0$, which again leads to a contradiction.  Therefore the eigenspace of $c$ must be one-dimensional.  Since the eigenspace must contain a $\sigma$-symmetric or $\sigma$-skew symmetric vector, and $\sigma$-skew symmetric vectors satisfy $v(0)=0$, we also see that the eigenspace has a $\sigma$-symmetric eigenvector.
\end{proof}

\begin{lem}
Let $\mu$ be the measure associated with the Hermitian Jacobi matrix $A$.  Then the support of the absolutely continuous component of $\mu$ is $[0,1]$.
\end{lem}
\begin{proof}
The support of the absolutely continuous component of $\mu$ is equal $\sigma_c(A)$.  Furthermore $\sigma_c(A)=\sigma_c(\wt B^2)$ and $\sigma_c(\wt B^2)=\sigma_c(\wt B)^2$.  Moreover, $\wt B$ differs from the Chebyshev operator
$$C = \left[\begin{array}{ccccccc}
0     & 1/2    & 0   &  0  & \dots\\
1/2   &   0    & 1/2 &  0  & \dots\\
0     & 1/2    &  0  & 1/2 & \dots\\
0     &   0    & 1/2 &  0  & \dots\\
\vdots &  \vdots  & \vdots & \vdots &  \ddots\\
\end{array}\right].$$
by a compact operator $K$, ie. $\wt B = C+K$.  Hence $\sigma_c(\wt B)=\sigma_c(C) = [-1,1]$, and it follows that $\sigma_c(A) = \sigma_c(\wt B)^2 = [0,1]$.
\end{proof}

\begin{prop}\label{Jacobi connection}
Let $c\in\bbc$ be non-real.  Then $c$ is an eigenvalue of Equation \ref{Rayleigh equation} if and only if $c$ is an element of the discrete spectrum of $A$.
\end{prop}
\begin{proof}
Suppose $c^2\in\sigma_d(A)$.  Then there exists $v\in\ell^2(\bbn)$ with $\vec v = [v(0),v(1),v(2),\dots]^T$ an eigenvector of $J$ with eigenvalue $c$.  Define $v(-j) = v(j)$ and 
$$u(j)
= \left\lbrace\begin{array}{cc}
v((j-1)/2)                              & \text{$j\geq 0$ odd}\\
(z_j/c)v(j/2)+(z_{j-1}/c)v(j/2-1) & \text{$j\geq 0$ even}\\
u(-j) & j < 0
\end{array}\right.$$
Then $[\dots,u(-2),u(-1),u(0),u(1),\dots]^T$ is an eigenvector of $B$.  Moreover this extends $v$ to an element of $\ell^2(\bbz)$.  It follows that $\wh f(\ell)$ satisfies Equation \ref{cosine problem} for $\wh f(0) = 0$ and for $\ell\neq 0$,
$$\wh f(\ell) = \frac{u(\ell)}{(\ell^2+k^2-1)^{1/2}(\ell^2 + k^2+1/\Bu)^{1/2}}.$$
Moreover, since $\wh f(\ell)$ is the product of two functions in $\ell^2(\bbz)$ we have that $\wh f(\ell)$ is in $\ell^2(\bbz)$.

To prove the converse, suppose that $c$ is an eigenvalue of Equation \ref{Rayleigh equation}, and let $\wh f(\ell)$ be the Fourier coefficients of the associated eigenfunction $f(y)$.  Then $c^2$ is an eigenvalue of the Jacobi matrix $A$, whose eigenvector is $[v(1),v(3),v(5),\dots]^T$ for $v(j) = \wh f(j)w(j)$ where
$$w(j) = (\ell^2+k^2-1)^{1/2}(\ell^2 + k^2 + 1/\Bu)^{1/2}.$$
However, the eigenspace of $A$ with eigenvalue $c^2$ is exactly $[p_0(c^2),p_1(c^2),p_2(c^2),\dots]$
where $p_0(x),p_1(x),p_2(x),\dots$ is the sequence of polynomials associated to $J$.  This implies that for some constant $K$
$$\wh f(2j+1)w(j) = Kp_{j}(c^2),\ \ \text{for}\ j\geq 0$$
and therefore that $[p_0(c^2)/w(0),p_1(c^2)/w(1),\dots]\in\ell^2(\bbn)$.
The continuous support of the measure $\mu$ associated to $A$ is contained on the positive real axis.  Therefore if $c^2\notin\sigma_d(A)$ then since $c^2$ is not a positive real number $c^2\notin\sigma(A)$.  The growth rate of $p_n(x)$ for $x\notin\supp(\mu)$ is exponential in $n$ by Theorem \ref{growth rate}, and since $w(n)$ has polynomial growth, this contradicts the possibility that $[p_0(c^2)/w(0), p_1(c^2)/w(1),\dots]$ is in $\ell^2(\bbn)$.  This completes the proof.
\end{proof}

\begin{prop}
Let $A$ be as above, and let $p_0(x),p_1(x),\dots$ be the sequence of polynomials associated to $A$, and let $r_{n1} < r_{n2} < \dots < r_{nn}$ be the roots of $p_n(x)$.  If $|k|<1$, then for all $n$ large enough, the polynomial $p_n(x)$ has exactly one negative root $r_{n1}$.  The spectrum of $A$ has exactly one negative value $r$, where $r$ is the limit of the monotone decreasing sequence $r_{11},r_{21},r_{31},\dots$.
\end{prop}
\begin{proof}
Fix $k$ with $|k|<1$.  Since $A$ is Hermitian, its spectrum is a subset of $\bbr$.  One may verify that $p_n(x)$ has a negative root for $n$ large enough.  If $p_n(x)$ has more than one negative root for some $n$, then by Corollary \ref{interlacing corollary}, then $\sigma(A)$ will have a negative limit point, and thus $\sigma(A)$ will have infinitely many negative elements.  Since the continuous part of the spectrum of $A$ is positive, this means that $\sigma_d(A)$ is infinite.  This implies that Rayleigh's equation has infinitely many eigenvalues $c$ for this value of $k$.

However, the number of eigenvalues of Rayleigh's equation is finite, as can be seen by the fact that the eigenvalues of $A$ correspond to the poles of the Stieltjes transform $m(z)$ of the measure $\mu$ associated to $A$.  Since the function $m(z)$ is meromorphic, its poles are discrete.  Therefore by the centered Howard's semicircle theorem, there are finitely many of them.  Alternatively for $\Bu=\infty$, one may use the result of Howard that the number of unstable modes is bounded by the number of inflection points of the background profile $u^b(y)$ \cite{howard1964number}.  Therefore $p_n(x)$ has at most one negative eigenvalue for each $n$.  Thus $p_n(x)$ has exactly one negative root $r_{n,1}$ for $n$ large enough, and by root interlacing (Proposition \ref{root interlacing}) $r_{n,1}$ is monotone decreasing.  This proves the proposition.
\end{proof}

We now turn to the proof of Theorem \ref{main theorem}, as stated in the introduction.  We have essentially proved it in the previous two propositions.
\begin{proof}[Proof of Theorem \ref{main theorem}]\mbox{}
\begin{enumerate}[(a)]
\item  If $|k|\geq 1$, then $B$ is a Hermitian matrix, and therefore the eigenvalues are all real.  Since the eigenvalues of $B$ determine the eigenvalues of Rayleigh's equation, this proves (a).
\item  This is a restatement of the conclusion of the previous proposition.
\item  This is a restatement of the conclusion of the previous proposition.
\item  This follows from Proposition \ref{Jacobi connection}.
\end{enumerate}
\end{proof}

\section{Numerical Results}
\begin{figure}[htp]
\begin{center}
\includegraphics[width=\linewidth]{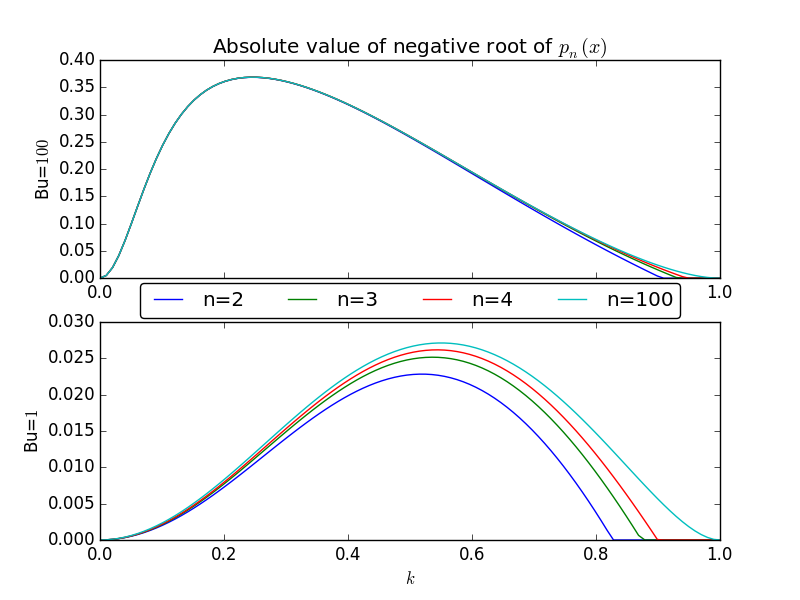}
\end{center}
\caption{The absolute value of the unique negative root $r_{n1}$ of $p_n(x)$ vs $k$ for various values of $n$ at Burger number $1$ and $100$.  As proved above, $-r_{n1}$ is monotone increasing.  Note that the rate of convergence is slower for smaller values of $\Bu$ and slower for larger wave number $k$.  Note that in the limit of large $n$, a negative root of $p_n(x)$ exists for all $0<k<1$}
\end{figure}
\begin{figure}[htp]
\begin{center}
\includegraphics[width=\linewidth]{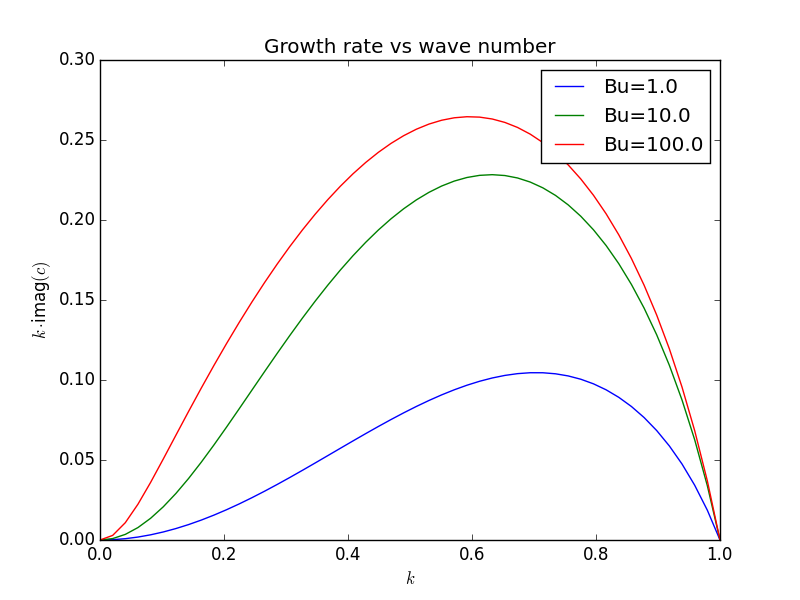}
\end{center}
\caption{The growth rate vs. wave number for various values of the Burger number.  Note that as the Burger number decreases, so too does the growth rate.}
\end{figure}
In this section we will provide a numerical verification of the result of Theorem \ref{main theorem}.  We will also provide a demonstration of the change in the behavior of the growth rate curve as a function of the Burger number $\Bu$.  The numerical calculations were carried out in python using numpy for the root calculations.
\section{Conclusions}
In this paper we explore the linear stability of the QG shallow water equation for a shear cosine profile with periodic boundary conditions, and derive bounds for the growth rate of the instabilities.  We relate the instabilities we see to the roots of a sequence of orthogonal polynomials, and our computational results verify our findings.

There are many unresolved questions that we would like to answer in the future, some of which we list here:
\begin{enumerate}[(a)]
\item  Is there an exact expression for the measure associated to our polynomials?
\item  Rayleigh's equation with the cosine profile $u^b(y)$ can be transformed into a Heun differential equation via the substitution $z = \cos(y)$, and therefore the associated Heun functions may be used as generating functions for our polynomials.  Do our polynomials comprise a ``nice" basis for the expansion of the associated Heun functions?
\item  Is there an analytic expression for growth rates?  Or else, can we establish estimates of the error in our approximation?
\end{enumerate}
We would also like to relate the linear instability we calculated here to the linear instability of the cosine profile for the full shallow water equations.

\section*{Acknowledgement}
The author is grateful for a summer research internship at Los Alamos National Laboratory, where part of the research for this paper was conducted.

\appendix
\section{Derivation of the Quasi-Geostrophic Shallow Water Equation}
\subsection{Shallow Water Equations}
In this paper, we consider the inviscid shallow water equations in a doubly periodic domain with a flat bottom.  The nondimensional form of the shallow water equations is given by
\begin{align}\label{shallow water equation}
\vec u_t + \vec u\cdot\grad \vec u = \frac{1}{\Ro}\vec u \times \hat{z} - \frac{1}{\Ro}\grad \eta\\
\eta_t  + \left(\frac{\Bu}{\Ro} + \eta\right)\grad \cdot \vec u + \vec u\cdot(\grad\eta) = 0
\end{align}
where $\vec u$ is the velocity field, $\eta$ is the free surface height, $\Ro = U/(fL)$ is the Rossby number, $\Fr = U/(gH)$ is the Froude number, $\Bu = (\Ro/\Fr)^2$ is the Burger number, $U$ is the characteristic magnitude of the velocity field, $H$ is the mean depth, $L$ is the characteristic length scale, $f$ is the Coriolis frequency, and $g$ is the (reduced) gravitational acceleration.

\subsection{Shallow Water QG Equation}\label{QG derivation}
In mid-latitude regions of the ocean or atmosphere at length scales relevent to geophysical flows, the Rossby number is typically small $\Ro \ll 1$.  In this situation, solutions to the shallow water equations are predominantly geostrophically balanced, eg. $\vec u\times\hat z = \grad\eta + \sheaf O(\Ro)$.  With this in mind, we next derive a geostrophically balanced model which approximates the shallow water equations in the limit of small Rossby number.

Consider a perturbation expansion in terms of the Rossby number $\Ro$ to obtain a balanced model for the evolution.  We expand the velocity and free surface height fields as
$$\vec u = \vec u^0 + \vec u^1\Ro + u^2\Ro^2 + \dots,$$
$$\eta = \eta^0 + \eta^1\Ro + \eta^2\Ro^2 + \dots.$$
Inserting this back into the shallow water equations and comparing similar powers of $\Ro$, we obtain two equations describing the leading balance for small Rossby number:
\begin{equation}\label{leading balance}
\grad\cdot \vec u^0 = 0,\ \ \ \grad\eta^0 = \vec u^0\times \hat z,
\end{equation}
as well as the following relations for all $i\geq 0$
\begin{align}\label{aux equations}
\vec u^i_t + \vec u^i\cdot\grad\vec u^i = \vec u^{i+1}\times \hat z - \grad\eta^{i+1},\\
\eta^i_t + \Bu\grad\cdot\vec u^{i+1} = 0.
\end{align}

To understand Equations (\ref{leading balance}) and (\ref{aux equations}) better, we introduce the following decomposition of the velocity field into a geostrophically balanced and imbalanced component:
$$\vec u^i = \vec u^i_0 + \vec u^i_1, \ \ \vec u^i_0 = \wh z\times \grad\eta^i,\ \vec u^i_1 = \vec u^i-\vec u^i_0.$$
Using this decomposition, Equations (\ref{leading balance}) simply say that $u^0_1 = 0$.  In other words, to leading order in $\Ro$ the velocity field is geostrophically balanced.

Next, note that
$$\vec u_1^i\times \wh z = (\vec u^i\times \wh z)-\grad\eta^i,$$
and therefore
\begin{equation}\label{evolution closed by ageostrophy}
\vec u^i_t + \vec u^i\cdot\grad\vec u^i = \vec u_1^{i+1}\times\wh{z}.
\end{equation}
Equation \ref{evolution closed by ageostrophy} shows that the time evolution of $\vec u^i$ is determined by $\vec u^i$ and the time evolution of the ageostrophic $\vec u^i_1$ field.

We next look at the field $\vec u^i_1$ more closely, examining in particular its divergence and curl.  Combining Equation \ref{evolution closed by ageostrophy} with the last of the four equations, we obtain the Helmholtz decomposition for $\vec u^{i+1}_1\times \wh{z}$:
\begin{equation}\label{helmholtz decomposition}
\vec u^{i+1}_1\times\wh{z} = -\grad p^i + \frac{\Delta^{-1}}{\Bu}(\grad \eta_t^i\times\wh{z}).
\end{equation}
where here $p^i$ is a pressure term given by
$$p^i =\frac{\Delta^{-1}}{\Bu}\eta_{tt}^{i-1} - \Delta^{-1}(\grad\cdot((\vec u^i\cdot\grad)\vec u^i)),$$
where $\eta_{tt}^{-1} := 0$.  Thus if we know $\eta^i$ and $\vec u^i$ for $i\leq n$, then we may obtain $\vec u^{n+1}_1$.

In the special case that $n=0$, Equation \ref{leading balance}, allows us to determine $\eta^0$ from $\vec u^0$.  Hence by combining Equation (\ref{helmholtz decomposition}) and Equation (\ref{evolution closed by ageostrophy}), we obtain a closed equation for the time evolution of $\vec u^0$, which we refer to as the shallow water quasi-geostrophic equation
\begin{align}\label{shallow water qg equation}
\left(1 - \frac{\Delta^{-1}}{\Bu}\right)\vec u^0_t + \vec u^0\cdot\grad\vec u^0 = -\grad p^0\\
p^0 = - \Delta^{-1}(\grad\cdot((\vec u^0\cdot\grad)\vec u^0)).
\end{align}
Note that by taking the curl of the above equation, we obtain the usual potential vorticity form of the QG shallow water equation
$$q^0_t + \vec u^0\cdot\grad q^0 = 0$$
$$q^0 = (\Delta-1/\Bu)\psi,\ \ \vec u^0 = -\grad\psi\times\hat z.$$

\section{Shear Instabilities and Rayleigh's Equation}\label{Rayleigh derivation}
The quasi-geostrophic shallow water (QG) equation describes the motion of a vertically homogeneous fluid in a rotating reference frame for fixed Burger number $\Bu$ in the limit of small Rossby number $\Ro$.  The QG equation is given by
\begin{equation}\label{qg equation}
q_t + J(\psi,q) = 0
\end{equation}
where here $q = q(x,y,t)$ is the potential vorticity, $\psi = \psi(x,y,t)$ is the stream function, and $J$ is the Jacobian
$$J(\psi,q) = \psi_xq_y - \psi_yq_x.$$  
The potential vorticity and the stream function are related to each other by
$$q = \grad^2\psi - \psi/\Bu.$$

We will consider the linear stability of solutions to the QG equation in a domain satisfying (normalized) periodic boundary conditions in the $y$ direction
$$\psi(x,y+2\pi) = \psi(x,y).$$
To obtain an equation for the linear stability of a given solution $\psi = \psi^b$ of \ref{qg equation}, we consider a solution of \ref{qg equation} of the form $\psi = \psi^b + \psi^p$, where $\psi^b$ is a base state solution and $\psi^p$ is a perturbation.  Inserting this back into the equation and ignoring quadratic terms in the perturbation, we obtain the linearized QG shallow water equation
$$q_t^p + J(\psi^b,q^p) + J(\psi^p,q^b) = 0.$$

In the case of a shear instabilities, the background solution is homogeneous in one of the directions, which we take to be the $x$-direction.  Then the background state is $\psi^b = -\int u^b(y) dy$ for some background velocity profile $u^b(y)$.  Then the linear equation reduces to
$$q_t^p + (q^b)'\psi^p_x - (\psi^b)'q^p_x = 0.$$
Rewriting this in terms of stream functions only, this says:
$$\left(\Delta-\frac{1}{\Bu}\right)\psi_t^p + (\psi^b)'''\psi^p_x - (\psi^b)'\Delta\psi^p_x = 0.$$
Since the coefficients of the above differential equation are constant in $x$ and $t$, it makes sense to look for solutions of the form
$$\psi^p(x,y,t) = e^{ik(x - c t)}f(y),$$
for some unknown function $f(y)$.  Substituting this in, we obtain a modified form of Rayleigh's Equation \ref{Rayleigh equation}:
$$f''(y)-\left(k^2 + \frac{u^b(y)''- c/\Bu}{u^b(y)-c}\right)f(y) = 0.$$
This differs from the usual Rayleigh equation in the inclusion of the $\Bu$ term.  As $\Bu$ increases however, this results in the usual form of Rayleigh's equation.

\section*{Acknowledgement}

\bibliographystyle{plain}
\bibliography{jacobi}

\end{document}